\documentclass[a4paper]{article}
\usepackage{a4wide}
\usepackage[UKenglish]{babel}
\usepackage{graphicx,subfigure,caption}
\graphicspath{{./}{figures/}}
\usepackage[colorlinks=true,linkcolor=blue,citecolor=red]{hyperref}
\usepackage{amssymb,amsthm,empheq,bbold}
\usepackage[inline]{enumitem}
\usepackage{authblk}
\usepackage{circledsteps}
\usepackage{algorithm,algpseudocode}

\newcommand{\abs}[1]{\left\vert#1\right\vert}
\newcommand{\ave}[1]{\left\langle#1\right\rangle}
\newcommand{\E}{\mathbb{E}}
\newcommand{\cE}{\mathcal{E}}
\renewcommand{\bf}{\mathbf{f}}
\newcommand{\cG}{\mathcal{G}}
\newcommand{\cI}{\mathcal{I}}
\newcommand{\indeg}[1]{\operatorname{indeg}(#1)}
\newcommand{\bM}{\mathbf{M}}
\newcommand{\cO}{\mathcal{O}}
\newcommand{\ones}{\mathbf{1}}
\newcommand{\outdeg}[1]{\operatorname{outdeg}(#1)}
\newcommand{\Prob}[1]{\operatorname{Prob}(#1)}
\newcommand{\R}{\mathbb{R}}
\newcommand{\bv}{\mathbf{v}}
\newcommand{\cV}{\mathcal{V}}
\newcommand{\bw}{\mathbf{w}}
\newcommand{\cW}{\mathcal{W}}

\newtheorem{lemma}{Lemma}[section]
\newtheorem{theorem}[lemma]{Theorem}
\theoremstyle{definition}\newtheorem{definition}[lemma]{Definition}
\theoremstyle{remark}\newtheorem{remark}[lemma]{Remark}
\theoremstyle{remark}\newtheorem{example}[lemma]{Example}

\allowdisplaybreaks

\title{Network-based kinetic models: Emergence of a statistical description of the graph topology}
\author[$\ast$,$\ddagger$]{M. Nurisso}
\author[$\dag$]{M. Raviola}
\author[$\ast$]{A. Tosin}
\affil[$\ast$]{{\footnotesize Department of Mathematical Sciences ``G. L. Lagrange'', Politecnico di Torino, Italy}}
\affil[$\dag$]{{\footnotesize CSQI Chair, \'{E}cole Polytechnique F\'{e}d\'{e}rale de Lausanne, Switzerland}}
\affil[$\ddagger$]{{\footnotesize CENTAI Institute, Turin, Italy}}
\date{}

\begin{document}
\maketitle
    
\begin{abstract}
In this paper, we propose a novel approach that employs kinetic equations to describe the collective dynamics emerging from graph-mediated pairwise interactions in multi-agent systems. We formally show that for large graphs and specific classes of interactions a statistical description of the graph topology, given in terms of the degree distribution embedded in a Boltzmann-type kinetic equation, is sufficient to capture the collective trends of networked interacting systems. This proves the validity of a commonly accepted heuristic assumption in statistically structured graph models, namely that the so-called connectivity of the agents is the only relevant parameter to be retained in a statistical description of the graph topology. Then we validate our results by testing them numerically against real social network data.
\medskip


\medskip

\end{abstract}

\section{Introduction}
In recent years, kinetic models have gained great popularity as convenient tools to study interacting multi-agent systems~\cite{pareschi2013BOOK}, which constitute the modelling paradigm for various socio-economic applications. Since such systems feature interconnected agents, a notion of \textit{graph} is often naturally required in the models. A prominent prototype is opinion dynamics in social networks, where agents interact only with their own contacts, i.e. their first neighbours in the graph modelling the social network. The vertices of such a graph are the agents while the edges describe the connections among them. The problem is that network-based models become quickly complex as the number of vertices grows, thereby posing challenges from the point of view of both mathematical analysis and computation. It is therefore natural to look for \textit{statistical} limit descriptions of the graph connections emerging when the size of the graph tends to infinity.

To this purpose, the notion of \textit{graphon} has been introduced in graph theory~\cite{lovasz2012BOOK}. A graphon is a suitable limit of a sequence of graphs of growing size conceived so as to represent large networks by a continuous model. Informally, given a graph $\cG_N$ with, say, $N\in\mathbb{N}$ vertices one first associates with the adjacency matrix $\bM_N\in\R^{N\times N}$ of $\cG_N$ a piecewise constant function $W_N$, which reproduces the entries of $\bM_N$ on a $N\times N$ discretisation of $[0,\,1]\times [0,\,1]\subset\R^2$ in sub-squares of side length $1/N$. Next, one takes the limit $N\to\infty$ of the sequence $\{W_N\}_{N\in\mathbb{N}}$, under a suitable notion of convergence, to possibly get a limit function $W:[0,\,1]^2\to [0,\,1]$, the so-called graphon, which describes the connections of the infinite-size limit graph.

In this paper, we take inspiration loosely from these ideas to incorporate a statistical continuous description of graph connections in Boltzmann-type kinetic equations used to represent statistically the collective dynamics of multi-agent systems. Our main contribution is that, under the classical molecular chaos hypothesis, we obtain kinetic equations able to account for the heterogeneous structure of the connections among the agents. We show that, for particular classes of binary interactions, these equations rely only on a statistical description of the graph topology. In particular, we formally prove that all information about the adjacency matrix of the graph may be lumped in the notion of \textit{degree distribution} of the graph, i.e. the statistical distribution of the numbers of incoming and outgoing edges of the vertices.

In more detail, our approach develops along the following line. The main idea is to augment the state of each agent by including in it, besides the variable, say $v$, characterising the interaction dynamics, two additional variables representing the incoming and outgoing degrees of the agent. Hence, the kinetic distribution function on the augmented state space describes the distribution of agents possessing a certain characteristic variable $v$ plus a certain number of incoming and outgoing connections. Passing formally to the limit of an infinite number of agents, i.e. of vertices of the graph, we show that this kinetic description converges, at least for certain classes of binary interactions, to a classical Boltzmann-type equation defined on the augmented state space and whose interaction kernel carries the information about the graph degrees. These therefore affect the rate of binary interactions among the agents, which turns out to be all the relevant information about agent connections from a statistical point of view.

It is worth recalling that previous works in this and related contexts approach the problem of upscaling particle descriptions to aggregate descriptions of networked interactions in different ways. Without claiming to be exhaustive, and confining ourselves to particularly recent contributions, we mention~\cite{burger2021network}, where networks are specified by linking structural variables to the agents and introducing interaction rates depending on them;~\cite{burger2021kinetic}, in which the joint evolution of the agents' states and the network itself is considered;~\cite{coppini2020law,delattre2016note}, in which particular classes of random graphs are studied.

After this introduction, the paper is organised as follows. In Section~\ref{sect:kinetic_graph-mediated} we derive kinetic equations on graphs starting from graph-mediated particle interactions. Here, the approach is technically reminiscent of that proposed in~\cite{loy2021KRM,loy2021MBE}, in particular for the fact that the graph is still a finite-size one and the connections among the vertices are described in detail by the adjacency matrix. A distribution function of the characteristic variable $v$ is associated with every vertex, viz. agent, and a system of coupled kinetic equations is derived, the coupling being dictated by the adjacency matrix of the graph (cf. Figure~\ref{fig:graph_illustration}a). In Section~\ref{sect:stat_description} we show that from this setting a single kinetic equation, defined on the aforementioned augmented state space and which incorporates a degree-based continuous description of the graph connections, naturally emerges in the limit of an infinite number of vertices/agents if one considers a very special class of binary interactions among the agents, that we call \textit{polarised memory} interactions. In Section~\ref{sect:equiv_boltz} we prove that the limiting kinetic equation so obtained can be recast in the form of a classical Boltzmann-type equation, whose interaction kernel features a precise dependence on the incoming and outgoing degrees of a generic representative vertex/agent of the system. In Section~\ref{sect:general_interactions} we investigate to what extent this result can be extended to a more general class of binary interactions, namely that of \textit{separable interactions} which includes, as a special case, linear interactions. In Section~\ref{sect:approximations_of_M} we prove that a quite natural rank-one approximation of the adjacency matrix of any graph allows for the extension of the results of the previous sections to arbitrary interaction rules. In Section~\ref{sect:numerics} we validate our theoretical findings by comparing numerically the dynamics produced by the original graph-mediated particle interactions and the solution to our Boltzmann-type equation, using data of user connections coming from a real social network. Finally, in Section~\ref{sect:conclusions} we draw some conclusions and we briefly sketch further possible research directions.

\begin{figure}[!t]
\centering
\includegraphics[width=\linewidth]{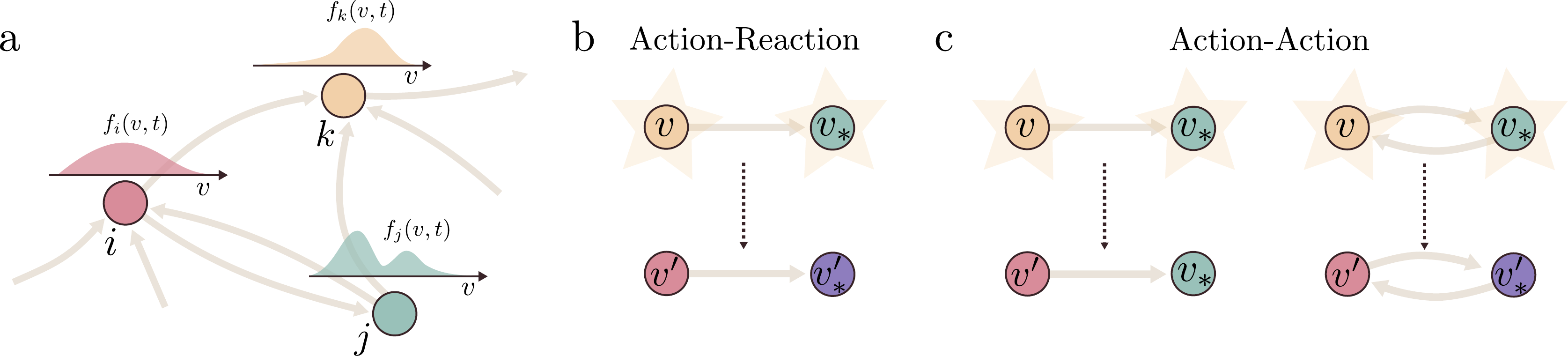}
\caption{\textbf{a.} Graphical representation of the interaction framework considered in this work. Each agent is identified with a vertex in a directed graph and is characterised by a probability distribution of their state which evolves in time. \textbf{b.} In an \textit{action-reaction} interaction between agents $i,\,j\in\cI$ connected by the edge $(i,\,j)\in\cE$ the states $v$, $v_\ast$ of both agents are updated. \textbf{c.} In an \textit{action-action} interaction, the state $v_\ast$ of agent $j$ is updated only if $(j,\,i)\in\cE$.}
\label{fig:graph_illustration}
\end{figure}

\section{Kinetic description of graph-mediated interactions}
\label{sect:kinetic_graph-mediated}
\subsection{Preliminaries}
Since social networks are among the most prominent examples of graphs encountered in social contexts, and various types of dynamics, often involving user opinions, take place on them, in the following we will keep as reference an application in the realm of opinion formation.  We will call \textit{opinion} a variable $v\in\cO\subseteq\R$, which is part of the microscopic state of a generic agent. However, it should be noted that the developments that follow are completely general and can be applied to other contexts as well. 
 
The opinions of the agents evolve because of interactions with other connected agents. The fundamental assumption we make, as usual in collisional kinetic theory, is that only binary interactions are relevant. In other words, we postulate that interactions among three or more agents are much rarer than those between two agents, so that their effect can be neglected.
 
Let us consider a generic representative agent, whose microscopic state is described by a stochastic process $(X,\,V_t)_{t\geq 0}$. In more detail, $X\in\cI$ is the location of the agent on a graph $\cG=(\cI,\,\cE)$, $\cI$ being the set of vertices and $\cE$ the set of edges of $\cG$. In practice, every agent is a vertex of $\cG$.  We assume that the graph is static, i.e. that connections among the agents do not change in time. Conversely, $V_t:\Omega\to\cO$ is a random variable from an abstract sample space $\Omega$ to the space of the opinions $\cO$ denoting the opinion of the agent at time $t\geq 0$. This random variable evolves in time due to binary interactions with other agents mediated by the connections encoded in $\cE$, thereby producing a stochastic process $\{V_t,\,t\in [0,\,+\infty)\}$. Overall, we describe statistically the microscopic state $(X,\,V_t)$ of the agent by means of a probability measure $f=f(x,v,t)$, which we understand as discrete in $x\in\cI$ and generically continuous in $v\in\cO$. Hence $f$ can be given the form
\begin{equation}
	f(x,v,t)=\frac{1}{N}\sum_{i\in\mathcal{I}}f_i(v,t)\otimes\delta(x-i),
	\label{eq:f}
\end{equation}
where $N=\abs{\cI}$ is the total number of agents, viz. vertices, of the graph while $\delta(\cdot)$ denotes the Dirac delta distribution centred at the origin. Moreover, $f_i=f_i(v,t):\cO\times [0,\,+\infty)\to\R_+$ is the probability density of the opinion $V_t$ of the agent $X=i$. See Figure~\ref{fig:graph_illustration}a. We require
$$ 
\int_\cO f_i(v,t)\,dv=1, \qquad \forall\,t\geq 0,\ \forall\,i\in\cI, 
$$
which implies consistently $\int_\cI\int_\cO f(x,v,t)\,dv\,dx=1$ at all times. Notice that then
$$ \Prob{X=i}=\frac{1}{N}, \qquad \forall\,i\in\cI, $$
meaning that every agent has the same probability to be sampled for an interaction. This corresponds to an \textit{a priori} uniform importance of each agent within the graph. Of course, one may reasonably expect that more connected agents -- the so-called \textit{influencers} -- contribute more to the formation of opinions on the social network. However, we believe that this should be an emergent feature of the model rather than an \textit{ad hoc} assumption.

\subsection{Interaction algorithms}
An interaction algorithm is a rule describing how agents interact in pairs and modify consequently their opinions over time. In particular, in a given time step $\Delta{t}>0$ we assume that an agent $(X,\,V_t)\in\cI\times\cO$ changes their opinion to $V_{t+\Delta{t}}\in\cO$ because of an interaction with another agent $(X^\ast,\,V^\ast_t)\in\cI\times\cO$ according to the following scheme:
\begin{equation}\label{eq:interaction_algorithm} 
V_{t+\Delta{t}}=(1-\Theta)V_t+\Theta V'_t, 
\end{equation} 
where $\Theta\in\{0,\,1\}$ is a random variable taking into account whether the interaction between the two agents actually produces ($\Theta=1$) or not ($\Theta=0$) an opinion change. Furthermore, $V_t'\in\cO$ is the new opinion acquired by agent $(X,\,V_t)$ in consequence of a successful interaction.

In more detail, we let
$$ \Theta\sim\operatorname{Bernoulli}\!\left(B(X,X^\ast)\Delta{t}\right), $$
meaning that the probability of a successful interaction is proportional to the interaction time step $\Delta{t}$ through an \textit{interaction kernel} $B=B(X,X^\ast)$, which encodes the information about the edges of the graph, viz. the connections among the agents. Specifically, we assume:
\begin{equation}
	B(X,X^\ast)=1\quad\text{if } (X,\,X^\ast)\in\cE, \qquad B(X,X^\ast)=0\quad\text{if } (X,\,X^\ast)\not\in\cE,
	\label{eq:B}
\end{equation}
where the ordered pair $(X,\,X^\ast)$ denotes the edge from vertex $X$ to vertex $X^\ast$. For consistency, we require $\Delta{t}\leq 1$, which imposes a limitation on the maximum admissible time step. However, we anticipate that this limitation will be unimportant when considering the continuous time limit $\Delta{t}\to 0^+$.

We model the post-interaction opinion as a random variable $V'_t:\Omega\to\cO$ depending in general on the pre-interaction opinions $V_t$, $V^\ast_t$ of the interacting agents:
$$ V'_t(\omega)=\Psi(V_t(\omega),V^\ast_t(\omega),\omega), \qquad \omega\in\Omega, $$
$\Psi:\cO^2\times\Omega\to\cO$ being a possibly stochastic given function.

Also the agent $(X^\ast,\,V^\ast_t)$ may simultaneously change opinion within the same binary interaction. The way in which this happens may vary depending on the characteristics of the interactions allowed by the social network.

\subsubsection{``Action-reaction'' interactions}
Assume the social network is such that an interaction of the agent in vertex $X$ with the agent in vertex $X^\ast$ necessarily implies also the interaction of the agent in vertex $X^\ast$ with the agent in vertex $X$. Technically, the ``forward'' interaction of $X$ with $X^\ast$ takes place only if the edge $(X,\,X^\ast)$ exists in $\cG$, i.e. if $(X,\,X^\ast)\in\cE$, whereas the ``backwards'' interaction of $X^\ast$ with $X$ may take place regardless of the existence of the edge $(X^\ast,\,X)$ in $\cG$. In other words, while the agent in vertex $X$ needs to decide actively to interact with the agent in vertex $X^\ast$ the latter simply reacts passively to the action of the former. See Figure~\ref{fig:graph_illustration}b.

Therefore, agent $(X^\ast,\,V^\ast_t)$ updates their opinion through a rule analogous to that of agent $(X,\,V_t)$:
$$ V^\ast_{t+\Delta{t}}=(1-\Theta)V^\ast_t+\Theta V^{\ast\prime}_t, $$
in particular with the \textit{same} random variable $\Theta$ whose law depends on $B(X,X^\ast)$ but not on $B(X^\ast,X)$. Furthermore, the interaction kernel $B$ need not be symmetric. The post-interaction opinion
$$ V^{\ast\prime}_t(\omega)=\Psi_\ast(V^\ast_t(\omega),V_t(\omega),\omega), \qquad \omega\in\Omega $$
is defined through a function $\Psi_\ast:\cO^2\times\Omega\to\cO$ possibly different from $\Psi$.

Summarising, for ``action-reaction'' interactions we consider the following general algorithm:
\begin{align}
	\begin{aligned}[c]
		V_{t+\Delta{t}} &= (1-\Theta)V_t+\Theta\Psi(V_t,V^\ast_t,\omega) \\
		V^\ast_{t+\Delta{t}} &= (1-\Theta)V^\ast_t+\Theta\Psi_\ast(V^\ast_t,V_t,\omega).
    \end{aligned}
    \label{eq:int.AR}
\end{align}
The interacting agents $(X,\,V_t)$, $(X^\ast,\,V^\ast_t)$ are sampled randomly and uniformly at each time step.

\subsubsection{``Action-action'' interactions}
Assume instead that the social network allows for an interaction of vertex $X$ with vertex $X^\ast$, and simultaneously of vertex $X^\ast$ with vertex $X$, only if each vertex is explicitly linked to the other, i.e. only if $(X,\,X^\ast),\,(X^\ast,\,X)\in\cE$. In other words, agents in either vertex need to take an action actively to interact with each other, none of them simply reacts to the action of the other. See Figure~\ref{fig:graph_illustration}c.

Then the random variable discriminating whether agent $(X^\ast,\,V^\ast_t)$ updates their opinion in the time step $\Delta{t}$ is in general \textit{different} from $\Theta$:
$$ \Theta_\ast\sim\operatorname{Bernoulli}\!\left(B(X^\ast,X)\Delta{t}\right) $$
and
$$ V^\ast_{t+\Delta{t}}=(1-\Theta_\ast)V^\ast_t+\Theta_\ast V^{\ast\prime}_t $$
with
$$ V^{\ast\prime}_t(\omega)=\Psi_\ast(V^\ast_t(\omega),V_t(\omega),\omega), \qquad \omega\in\Omega. $$
Again, the interaction kernel $B$ need not be symmetric.

Overall, for ``action-action'' interactions we consider the following general algorithm:
\begin{align}
	\begin{aligned}[c]
		V_{t+\Delta{t}} &= (1-\Theta)V_t+\Theta\Psi(V_t,V^\ast_t,\omega) \\
		V^\ast_{t+\Delta{t}} &= (1-\Theta_\ast)V^\ast_t+\Theta_\ast\Psi_\ast(V^\ast_t,V_t,\omega),
    \end{aligned}
    \label{eq:int.AA}
\end{align}
the interacting agents $(X,\,V_t)$, $(X^\ast,\,V^\ast_t)$ being sampled again randomly and uniformly at each time step.

\begin{remark}
\label{rem:undirected.directed}
If the interaction kernel $B$ is symmetric, i.e.
$$ B(X,X^\ast)=B(X^\ast,X), \qquad \forall\,X,\,X^\ast\in\cI, $$
then ``action-reaction'' interactions may be interpreted as ``action-action'' interactions over an undirected graph. Indeed, the symmetry of $B$ together with the presence of $\Theta$ in both interaction rules of algorithm~\eqref{eq:int.AR} implies that vertex $X^\ast$ is linked to vertex $X$ whenever the converse is true and that the edge connecting them is the same in both directions.

On the other hand, for a general non-symmetric interaction kernel $B$ ``action-action'' interactions may be regarded as binary interactions on a directed graph.
\end{remark}

\subsection{Derivation of kinetic equations}
\label{sect:derivation_kin_eq}
A kinetic description of algorithms~\eqref{eq:int.AR},~\eqref{eq:int.AA} amounts to evolution equations for the probability distributions $f_i$ of the opinion of the agents. To derive them, we adopt a classical procedure of the kinetic theory of multi-agent systems, see e.g.,~\cite[Appendix A]{fraia2020RUMI} or~\cite{pareschi2013BOOK}.

In the following, we develop explicit calculations for the case of ``action-reaction'' interactions, i.e. for the interaction algorithm~\eqref{eq:int.AR}. Afterwards, we will indicate the necessary modifications to treat also the case of ``action-action'' interactions with algorithm~\eqref{eq:int.AA}.

Let $\Phi=\Phi(x,v):\cI\times\cO\to\R$ be an arbitrary observable (test function), i.e. any quantity that can be computed out of the knowledge of the microscopic state of a generic representative agent of the system. 
From the first equation in~\eqref{eq:int.AR}, taking the expectation of the post-interaction observable, we have:
\begin{align*}
	\E[\Phi(X,V_{t+\Delta{t}})] &= \E\Bigl[\E\bigl[\Phi\bigl(X,(1-\Theta)V_t+\Theta\Psi(V_t,V^\ast_t,\omega)\bigr)\vert X,X^\ast\bigr]\Bigr] \\
    &= \E\!\left[\Phi(X,V_t)\left(1-B(X,X_\ast)\Delta{t}\right)+\Phi\bigl(X,\Psi(V_t,V^\ast_t,\omega)\bigr)B(X,X^\ast)\Delta{t}\right].
\end{align*}
Rearranging the terms and dividing both sides by $\Delta{t}$ gives 
$$ \frac{\E[\Phi(X,V_{t+\Delta{t}})]-\E[\Phi(X,V_t)]}{\Delta{t}}=\E\left[B(X,X^\ast)\bigl(\Phi(X,\Psi(V_t,V^\ast_t,\omega))-\Phi(X,V_t)\bigr)\right], $$
whence taking the limit $\Delta{t}\to 0^+$ yields formally
\begin{equation}
	\frac{d\E[\Phi(X,V_t)]}{dt}=\E\left[B(X,X^\ast)\bigl(\Phi(X,\Psi(V_t,V^\ast_t,\omega))-\Phi(X,V_t)\bigr)\right].
	\label{eq:expected}
\end{equation}
Similar calculations based on the second equation in~\eqref{eq:int.AR} produce
\begin{equation}
	\frac{d\E[\Phi(X^\ast,V^\ast_t)]}{dt}=\E\left[B(X,X^\ast)\bigl(\Phi(X^\ast,\Psi_\ast(V^\ast_t,V_t,\omega))-\Phi(X^\ast,V^\ast_t)\bigr)\right].
	\label{eq:expected.2}
\end{equation}

We now observe that both pairs $(X,V_t)$ and $(X^\ast,V^\ast_t)$ refer to a generic representative agent of the system, hence they have the same probability law. 
In particular, $\E[\Phi(X,V_t)]=\E[\Phi(X^\ast,V^*_t)]$ so that, summing~\eqref{eq:expected} and~\eqref{eq:expected.2}, we deduce
\begin{align*}
	\frac{d\E[\Phi(X,V_t)]}{dt}=\frac{1}{2}\E\Bigl[B(X,X^\ast)\Bigl(&\Phi(X,\Psi(V_t,V^\ast_t,\omega))+\Phi(X^\ast,\Psi_\ast(V^\ast_t,V_t,\omega)) \\
	&-\Phi(X,V_t)-\Phi(X^\ast,V^\ast_t)\Bigr)\Bigr]
\end{align*}
and finally, making use of the probability measure $f$ to compute the remaining expectations,
\begin{align}
	\resizebox{.94\hsize}{!}{$
	\begin{aligned}[b]
		&\frac{d}{dt}\int_\cI\int_\cO\Phi(x,v)f(x,v,t)\,dv\,dx= \\
		&=\int_{\cI^2}\int_\cO\int_\cO B(x,x_\ast)\frac{\ave{\Phi(x,v')+\Phi(x_\ast,v_\ast')-\Phi(x,v)-\Phi(x_\ast,v_\ast)}}{2}f(x,v,t)f(x_\ast,v_\ast,t)\,dv\,dv_\ast\,dx\,dx_\ast,
	\end{aligned}
	$}
	\label{eq:Boltzmann.f}
\end{align}
where we have denoted
\begin{equation}
	v'=\Psi(v,v_\ast,\omega), \qquad v_\ast'=\Psi_\ast(v_\ast,v,\omega)
	\label{eq:int_rules}
\end{equation}
for brevity and where $\ave{\cdot}$ denotes expectation with respect to the possible stochasticity of the functions $\Psi$, $\Psi_\ast$.

Notice that~\eqref{eq:Boltzmann.f} is required to hold for all test functions $\Phi$. As such, it is a \textit{weak} equation for the distribution function $f$. The corresponding strong form might be given but it is irrelevant for the sequel, therefore we neglect it.

\begin{remark}
Equation~\eqref{eq:Boltzmann.f} is written under the assumption of \textit{propagation of chaos}, meaning that any two potentially interacting agents are sampled independently of each other. This assumption is classically used e.g., in the Boltzmann-type kinetic theory to obtain a closed equation for the one-particle distribution function $f$, as it allows one to factorise the joint probability measure $f_2=f_2(x,x_\ast,v,v_\ast,t)$ of the interacting agents in the product $f(x,v,t)f(x_\ast,v_\ast,t)$.
\end{remark}

From~\eqref{eq:Boltzmann.f}, with a convenient choice of the test function $\Phi$, it is possible to recover a system of (weak) equations for the $f_i$'s. Let $\Phi(x,v)=\phi_i(x)\varphi(v)$, where $\phi_i:\cI\to\R$ is such that $\phi_i(i)=1$ while $\phi_i(x)=0$ for all $x\in\cI\setminus\{i\}$ and $\varphi:\cO\to\R$ is arbitrary. Then, plugging~\eqref{eq:f} into~\eqref{eq:Boltzmann.f} yields:
\begin{align}
        \resizebox{.93\hsize}{!}{$
	\begin{aligned}[b]
	    \frac{d}{dt}\int_\cO\varphi(v)f_i(v,t)\,dv &= \frac{1}{2N}\sum_{j\in\cI}B(i,j)\int_\cO\int_\cO\ave{\varphi(v')-\varphi(v)}f_i(v,t)f_j(v_\ast,t)\,dv\,dv_\ast \\
    		&\phantom{=} +\frac{1}{2N}\sum_{j\in\cI}B(j,i)\int_\cO\int_\cO\ave{\varphi(v_\ast')-\varphi(v_\ast)}f_j(v,t)f_i(v_\ast,t)\,dv\,dv_\ast, \quad i\in\cI.
    	\end{aligned}
        $}
    \label{eq:bg}
\end{align}
This equation can also be obtained, still under the assumption of propagation of chaos, from the kinetic equation derived in~\cite{burger2021network}, which stems from a BBGKY-type hierarchy. Moreover,~\eqref{eq:bg} can be conveniently recast in matrix form by introducing the vector-valued distribution function $\bf(v,t):=(f_i(v,t))_{i\in\cI}$ and the matrix $\bM:=(B(i,j))_{i,j\in\cI}\in\R^{N\times N}$:
\begin{align}
	\begin{aligned}[b]
	    \frac{d}{dt}\int_\cO\varphi(v)\bf(v,t)\,dv &= \frac{1}{2N}\int_\cO\int_\cO\ave{\varphi(v')-\varphi(v)}\bf(v,t)\odot\bM\bf(v_\ast,t)\,dv\,dv_\ast \\
    		&\phantom{=} +\frac{1}{2N}\int_\cO\int_\cO\ave{\varphi(v_\ast')-\varphi(v_\ast)}\bM^T\bf(v,t)\odot\bf(v_\ast,t)\,dv\,dv_\ast,
    \end{aligned}
    \label{eq:AR-f}
\end{align}
where $\odot$ denotes the Hadamard product and $\bM^T$ the transpose matrix of $\bM$. Note that $\mathbf{M}$ is the \textit{adjacency matrix} of $\cG$.

In the case of ``action-action'' interactions, owing to the second equation in~\eqref{eq:int.AA}, equation~\eqref{eq:expected.2} is replaced by
\begin{equation}\label{eq:Boltz_type_AA}  
\frac{d\E[\Phi(X^\ast,V^\ast_t)]}{dt}=\E\left[B(X^\ast,X)\bigl(\Phi(X^\ast,\Psi_\ast(V^\ast_t,V_t,\omega))-\Phi(X^\ast,V^\ast_t)\bigr)\right], 
\end{equation} 

which, added to~\eqref{eq:expected}, produces now
\begin{align*}
	\frac{d}{dt}\int_\cI\int_\cO &\Phi(x,v)f(x,v,t)\,dv\,dx= \\
	&=\frac{1}{2}\int_{\cI^2}\int_\cO\int_\cO B(x,x_\ast)\ave{\Phi(x,v')-\Phi(x,v)}f(x,v,t)f(x_\ast,v_\ast,t)\,dv\,dv_\ast\,dx\,dx_\ast \\
	&\phantom{=} +\frac{1}{2}\int_{\cI^2}\int_\cO\int_\cO B(x_\ast,x)\ave{\Phi(x_\ast,v_\ast')-\Phi(x_\ast,v_\ast)}f(x,v,t)f(x_\ast,v_\ast,t)\,dv\,dv_\ast\,dx\,dx_\ast
\end{align*}
in place of~\eqref{eq:Boltzmann.f} and finally
\begin{align}
	\begin{aligned}[b]
	    \frac{d}{dt}\int_\cO\varphi(v)\bf(v,t)\,dv &= \frac{1}{2N}\int_\cO\int_\cO\ave{\varphi(v')-\varphi(v)}\bf(v,t)\odot\bM\bf(v_\ast,t)\,dv\,dv_\ast \\
    		&\phantom{=} +\frac{1}{2N}\int_\cO\int_\cO\ave{\varphi(v_\ast')-\varphi(v_\ast)}\bM\bf(v,t)\odot\bf(v_\ast,t)\,dv\,dv_\ast,
    \end{aligned}
    \label{eq:AA-f}
\end{align}
in place of~\eqref{eq:AR-f}. Notice that if $B$ is symmetric then so is $\bM$ and hence~\eqref{eq:AR-f} and~\eqref{eq:AA-f} coincide.

\section{Statistical description of the connections}
\label{sect:stat_description}
System~\eqref{eq:bg}, or equivalently~\eqref{eq:AR-f}, describes the evolution of the probability density of the opinion of each user of the social network. Therefore, it is in general not easily amenable to analytical or numerical investigations, because one may reasonably expect that the total number $N$ of agents, hence of kinetic equations, is quite large in all realistic applications.

To get rid of the necessity to track, also at the kinetic level, individual agents, viz. vertices of $\cG$, one may look for the global opinion distribution on the social network, i.e. the $v$-marginal distribution of $f$ in~\eqref{eq:f}:
$$ F(v,t):=\int_\cI f(x,v,t)\,dx=\frac{1}{N}\sum_{i\in\cI}f_i(v,t)=\frac{1}{N}\ones^T\bf(v,t), $$
where $\ones^T=(1,\,\dots,\,1)\in\R^N$.

Let us focus preliminarily on ``action-reaction'' interactions. Premultiplying~\eqref{eq:AR-f} by $\frac{1}{N}\ones^T$ we get
\begin{equation}
	\frac{d}{dt}\int_\cO\varphi(v)F(v,t)\,dv=\frac{1}{2N^2}\int_\cO\int_\cO\ave{\varphi(v')+\varphi(v_\ast')-\varphi(v)-\varphi(v_\ast)}\bf^T(v,t)\bM\bf(v_\ast,t)\,dv\,dv_\ast
	\label{eq:F}
\end{equation}
which however is not a closed equation for $F$, because the right-hand side still requires the detailed knowledge of the connections of the graph and of the opinion distribution of each agent.

In order to draw self-consistent information from~\eqref{eq:F}, we restrict at first to classes of sufficiently simple interaction dynamics, yet capable of giving rise to non-trivial collective trends. To this purpose, we introduce the following

\begin{definition}
\label{def:polarised_memory}
We say that an interaction rule
$$ v'=\Psi(v,v_\ast,\omega) $$
has \textit{perfect memory} if $\Psi$ is constant w.r.t. $v_\ast$, so that the post-interaction opinion $v'$ of the agent with pre-interaction opinion $v$ is independent of the opinion $v_\ast$ of the other interacting agent.

Conversely, we say that the above interaction rule is \textit{memoryless} if $\Psi$ is constant w.r.t. $v$, so that the post-interaction opinion $v'$ of the agent with pre-interaction opinion $v$ depends only on the opinion $v_\ast$ of the other interacting agent (plus possibly independent stochastic effects).

We call \textit{polarised memory} interactions the class of perfect memory and memoryless interactions.
\end{definition}

\begin{remark}
In the context of opinion formation, the term \textit{memory} introduced in Definition~\ref{def:polarised_memory} may be understood as a synonym of \textit{conviction}. Specifically, a perfect memory interaction is one in which an agent has a so strong conviction that they disregard completely the opinion of the other agent when forming their own post-interaction opinion. This is, for instance, the case of opinion leaders. Conversely, a memoryless interaction is one in which an agent has a so weak conviction that their post-interaction opinion is influenced completely by the opinion of the other agent and is independent of their own pre-interaction opinion. This may happen e.g., among opinion followers.

In the following, however, we keep the term ``memory'' as it is somehow more general and, as such, more easily adaptable, in the abstract, also to applications different from opinion formation.
\end{remark}

\begin{example}\label{ex:sznajd}
In the Ochrombel simplification~\cite{ochrombel2001IJMP} of the Sznajd model~\cite{sznajd-weron2000IJMP} of opinion dynamics, the interaction rules are
$$ v'=v, \qquad v_\ast'=v. $$
They are both polarised memory interactions and, in particular, the first rule is a perfect memory one while the second rule is memoryless.
\label{ex:Ochrombel}
\end{example}

Motivated by Example~\ref{ex:Ochrombel}, let us focus on interaction rules~\eqref{eq:int_rules} with $\Psi=\Psi(v,\omega)$ (perfect memory) and $\Psi_\ast=\Psi_\ast(v,\omega)$ (memoryless). If this is the case, the right-hand side of~\eqref{eq:F} may be rewritten as
\begin{multline*}
	\frac{1}{2N^2}\int_\cO\int_\cO\ave{\varphi(v')+\varphi(v_\ast')-\varphi(v)-\varphi(v_\ast)}\bf^T(v,t)\bM\bf(v_\ast,t)\,dv\,dv_\ast= \\
	=\frac{1}{2N^2}\int_\cO\ave{\varphi(v')+\varphi(v_\ast')-\varphi(v)}\bf^T(v,t)\,dv\ \bM\ones-\ones^T\bM\ \frac{1}{2N^2}\int_\cO\varphi(v_\ast)\bf(v_\ast,t)\,dv_\ast.
\end{multline*}

Let
\begin{equation}
	\bw^-:=\bM^T\ones, \qquad \bw^+:=\bM\ones,
	\label{eq:w-_w+}
\end{equation}
be the vectors of the incoming ($-$) and outgoing ($+$) degrees, viz. the numbers of incoming and outgoing edges of each of the vertices of the graph $\cG$. 
Using them and taking advantage of the previous calculations, we rewrite~\eqref{eq:F} as
\begin{equation}
	\frac{d}{dt}\int_\cO\varphi(v)F(v,t)\,dv=\frac{1}{N^2}\int_\cO\ave{(\bw^+)^T\frac{\varphi(v')+\varphi(v_\ast')}{2}-\frac{(\bw^-)^T+(\bw^+)^T}{2}\varphi(v)}\bf(v,t)\,dv.
    \label{eq:ryan_equation}
\end{equation}

Now, the crucial point to obtain a kinetic formulation free from references to single vertices, viz. individual agents, is to augment the space of microscopic states to include also information on the connections of a generic representative vertex of the graph $\cG$. This way, such an information will be amenable to a statistical description. This may be difficult to do in general in~\eqref{eq:F}, while it is doable in~\eqref{eq:ryan_equation} because here the adjacency matrix $\bM$ of $\cG$ is, in a sense, lumped in $\bw^-$, $\bw^+$.

For $i\in\cI$, let $\indeg{i},\,\outdeg{i}\in\{0,\,\dots,\,N\}$ be the incoming and outgoing, respectively, degrees of vertex $i$. In other words, they are the $i$-th components of the vectors $\bw^-$, $\bw^+$, respectively. We introduce the function
$$ g_N=g_N(v,w^-,w^+,t):\cO\times\{0,\,\dots,\,N\}\times\{0,\,\dots,\,N\}\times\R_+\to\R_+ $$
meant as the probability mass function of the event that at time $t$ an agent with incoming degree $w^-$ and outgoing degree $w^+$ has opinion $v$. Owing to this definition, $g_N$ is linked to the $f_i$'s by the relationship
\begin{equation}
	g_N(v,w^-,w^+,t)=\frac{1}{N}\sum_{\substack{i\in\cI \\ \indeg{i}=w^-, \\ \outdeg{i}=w^+}}f_i(v,t)
    \label{eq:g}
\end{equation}
whence we deduce
\begin{align}
	\begin{aligned}[c]
	    F(v,t) &= \sum_{w^+=0}^{N}\sum_{w^-=0}^{N}g_N(v,w^-,w^+,t), \\
    		(\bw^-)^T\bf(v,t) &= N\sum_{w^+=0}^{N}\sum_{w^-=0}^{N}w^-g_N(v,w^-,w^+,t), \\
	    (\bw^+)^T\bf(v,t) &= N\sum_{w^+=0}^{N}\sum_{w^-=0}^{N}w^+g_N(v,w^-,w^+,t).
	\end{aligned}
	\label{eq:w.f}
\end{align}
Therefore, we may rewrite~\eqref{eq:ryan_equation} in terms of $g_N$ as
\begin{multline}
	\frac{d}{dt}\sum_{w^+=0}^{N}\sum_{w^-=0}^{N}\int_\cO\varphi(v)g_N(v,w^-,w^+,t)\,dv \\
	=\frac{1}{N}\sum_{w^+=0}^{N}\sum_{w^-=0}^{N}\int_\cO\ave{w^+\frac{\varphi(v')+\varphi(v_\ast')}{2}-\frac{w^-+w^+}{2}\varphi(v)}g_N(v,w^-,w^+,t)\,dv.
	\label{eq:to_be_continued}
\end{multline}
Let now
$$ \tilde{w}^-:=\frac{w^-}{N}, \qquad \tilde{w}^+:=\frac{w^+}{N} $$
be the normalised incoming and outgoing degrees of a generic representative vertex of $\cG$ and let us define
$$ \tilde{g}(v,\tilde{w}^-,\tilde{w}^+,t):=N^2g_N(v,N\tilde{w}^-,N\tilde{w}^+,t). $$
Notice that
$$ \tilde{w}^-,\,\tilde{w}^+\in\cW:=\left\{\frac{k}{N},\,k=0,\,1,\,\dots,\,N\right\} $$
and that the step between any two consecutive elements of $\cW$ is constant and equal to $\frac{1}{N}$. Therefore, introducing
$$ \Delta{\tilde{w}^-}=\Delta{\tilde{w}^+}:=\frac{1}{N}, $$
it results
$$ \sum_{\tilde{w}^+\in\cW}\sum_{\tilde{w}^-\in\cW}\int_\cO\tilde{g}(v,\tilde{w}^-,\tilde{w}^+,t)\,dv\,\Delta{\tilde{w}^-}\,\Delta{\tilde{w}^+}=
	\sum_{w^+=0}^{N}\sum_{w^-=0}^{N}\int_\cO g_N(v,w^-,w^+,t)\,dv=1, $$
which allows us to understand $\tilde{g}$ as a $\tilde{w}^-,\,\tilde{w}^+$-piecewise constant probability density function for all $N$.

Using $\tilde{g}$ to further elaborate~\eqref{eq:to_be_continued}, we may write
\begin{multline*}
	\frac{d}{dt}\sum_{\tilde{w}^+\in\cW}\sum_{\tilde{w}^-\in\cW}\int_\cO\varphi(v)\tilde{g}(v,\tilde{w}^-,\tilde{w}^+,t)\,dv\,\Delta{\tilde{w}^-}\,\Delta{\tilde{w}^+} \\
	=\sum_{\tilde{w}^+\in\cW}\sum_{\tilde{w}^-\in\cW}
		\int_\cO\ave{\tilde{w}^+\frac{\varphi(v')+\varphi(v_\ast')}{2}-\frac{\tilde{w}^-+\tilde{w}^+}{2}\varphi(v)}\tilde{g}(v,\tilde{w}^-,\tilde{w}^+,t)\,dv\,\Delta{\tilde{w}^-}\,\Delta{\tilde{w}^+},
\end{multline*}
which is a self-consistent equation for $\tilde{g}$. Assuming now that the size $N$ of the graph is large, i.e. letting $N\to\infty$, we transform $\tilde{w}^-$, $\tilde{w}^+$ into continuous variables ranging in the interval $[0,\,1]$ and from the previous equation we get formally (dropping the tildes for ease of notation)
\begin{multline}
	\frac{d}{dt}\int_0^1\int_0^1\int_\cO\varphi(v)g(v,w^-,w^+,t)\,dv\,dw^-\,dw^+ \\
	=\int_0^1\int_0^1\int_\cO\ave{w^+\frac{\varphi(v')+\varphi(v_\ast')}{2}-\frac{w^-+w^+}{2}\varphi(v)}g(v,w^-,w^+,t)\,dv\,dw^-\,dw^+.
	\label{eq:Boltzmann.g-AR}
\end{multline}

With~\eqref{eq:Boltzmann.g-AR} we have upscaled, in a formally rigorous way, a microscopic description based on single agents and their individual connections in the graph to a simplified aggregate description in which the graph appears only through the continuous statistical distribution of the normalised degrees of its vertices, at least in the case of interactions with polarised memory.

The same procedure applied to~\eqref{eq:AA-f} for the case of ``action-action'' interactions yields
\begin{multline}
	\frac{d}{dt}\int_0^1\int_0^1\int_\cO\varphi(v)g(v,w^-,w^+,t)\,dv\,dw^-\,dw^+ \\
	=\int_0^1\int_0^1\int_\cO\ave{\frac{w^+\varphi(v')+w^-\varphi(v_\ast')}{2}-w^+\varphi(v)}g(v,w^-,w^+,t)\,dv\,dw^-\,dw^+
	\label{eq:Boltzmann.g-AA}
\end{multline}
in place of~\eqref{eq:Boltzmann.g-AR}.

\begin{remark}
If the interaction kernel $B$ is symmetric then so is the adjacency matrix $\bM$ of $\cG$ and consequently $\bw^-=\bw^+$. In this case, we may describe the degree of a vertex of $\cG$ by a single variable $w:=w^-=w^+$ and replace the distribution function $g$ by
$$ \bar{g}(v,w,t):=\int_0^1g(v,w,w^+,t)\,dw^+, $$
so that both~\eqref{eq:Boltzmann.g-AR} and~\eqref{eq:Boltzmann.g-AA} reduce to
\begin{equation}
	\frac{d}{dt}\int_0^1\int_\cO\varphi(v)\bar{g}(v,w,t)\,dv\,dw
		=\int_0^1\int_\cO w\ave{\frac{\varphi(v')+\varphi(v_\ast')}{2}-\varphi(v)}\bar{g}(v,w,t)\,dv\,dw.
	\label{eq:Boltzmann.B_symm}
\end{equation}
This equation describes binary interactions, one of which is perfect memory while the other is memoryless (like in the Ochrombel model, cf. Example~\ref{ex:Ochrombel}), on an undirected graph, cf. Remark~\ref{rem:undirected.directed}.
\end{remark}

\section{Equivalent Boltzmann-type descriptions}
\label{sect:equiv_boltz}
In this section, we show that~\eqref{eq:Boltzmann.g-AR},~\eqref{eq:Boltzmann.g-AA} can be recovered as particular instances of general Boltzmann-type equations for a generic system of interacting agents, independently of the construction which takes into account explicitly the graph structure of the interactions. As we will see, to this purpose the key point is twofold: on one hand, it is necessary to characterise the microscopic state of the agents conveniently, in particular not only by their opinion $v$ but also by their (normalised) incoming and outgoing degrees $w^-$, $w^+$. On the other hand, it is fundamental to identify an appropriate expression of the collision kernel of the Boltzmann-type equation in terms of the degrees $w^-$, $w^+$ of a generic representative agent of the system.

The equivalence of~\eqref{eq:Boltzmann.g-AR},~\eqref{eq:Boltzmann.g-AA} with classical Boltzmann-type equations corroborates rigorously a commonly accepted heuristic assumption made in the literature, according to which networked interactions may be described statistically by invoking the concept of \textit{connectivity} of the agents, see e.g.,~\cite{he2023JAMP,loy2022PTRSA,toscani2018PRE}. In those contexts, the connectivity is understood as a measure of the number of connections of an agent in a (large) graph of agents.

\subsection{The case of ``action-reaction'' interactions}
``Action-reaction'' interactions, cf.~\eqref{eq:int.AR}, are binary interactions similar to those of the classical collisional kinetic theory. Therefore, we may refer to a classical Boltzmann-type kinetic equation, which for agents described by a generic microscopic state $\bv\in\cV\subseteq\R^d$ writes in weak form as (cf.~\cite{pareschi2013BOOK})
\begin{align}
	\begin{aligned}[b]
		\frac{d}{dt}&\int_\cV\Phi(\bv)g(\bv,t)\,d\bv \\
		&= \frac{1}{2}\int_{\cV^2}b(\bv,\bv_\ast)\ave{\Phi(\bv')+\Phi(\bv_\ast')-\Phi(\bv)-\Phi(\bv_\ast)}
				g(\bv,t)g(\bv_\ast,t)\,d\bv\,d\bv_\ast
    \end{aligned}
	\label{eq:std_Boltz-AR}
\end{align}
for every observable $\Phi:\cV\to\R$. The function $b:\cV^2\to\R_+$ is the collision kernel.

Now, let us choose $\bv=(v,\,w^-,\,w^+)$ with $\cV=\mathcal{O}\times [0,\,1]\times [0,\,1]$ and
\begin{equation}
	b(\bv,\bv_\ast)=\mu w^+w^-_\ast,
	\label{eq:b}
\end{equation}
where $\mu>0$ is a proportionality constant. Moreover, let us consider a $v$-dependent only observable, i.e. $\Phi(\bv)=\varphi(v)$. With the further assumption of polarised memory interactions,~\eqref{eq:std_Boltz-AR} becomes
\begin{align}
	\begin{aligned}[b]
		&\frac{d}{dt}\int_0^1\int_0^1\int_\cO\varphi(v)g(v,w^-,w^+,t)\,dv\,dw^-\,dw^+= \\
		&= \frac{\mu}{2}\int_{[0,\,1]^3}w^+w_\ast^-\left(\int_\cO\ave{\varphi(v')+\varphi(v_\ast')-\varphi(v)}g(v,w^-,w^+,t)\,dv\right)h^-(w_\ast^-)\,dw^-\,dw^+\,dw_\ast^- \\
		&\phantom{=} -\frac{\mu}{2}\int_{[0,\,1]^3}w^+w_\ast^-\left(\int_\cO\varphi(v_\ast)g(v_\ast,w_\ast^-,w_\ast^+,t)\,dv_\ast\right)h^+(w^+)\,dw^+\,dw_\ast^-\,dw_\ast^+,
	\end{aligned}
	\label{eq:std_Boltz-AR.1}
\end{align}
where
$$ h^-(w^-):=\int_0^1\int_\cO g(v,w^-,w^+,t)\,dv\,dw^+, \qquad h^+(w^+):=\int_0^1\int_\cO g(v,w^-,w^+,t)\,dv\,dw^- $$
are the marginal distributions of the (normalised) incoming and outgoing degrees, respectively, which do not change in time due to the assumption of static graph. Upon introducing the mean (normalised) incoming and outgoing degrees:
$$ m^-:=\int_0^1w^-h^-(w^-)\,dw^-, \qquad m^+:=\int_0^1w^+h^+(w^+)\,dw^+, $$
equation~\eqref{eq:std_Boltz-AR.1} can be written in the form
\begin{align}
	\begin{aligned}[b]
		\frac{d}{dt}&\int_0^1\int_0^1\int_\cO\varphi(v)g(v,w^-,w^+,t)\,dv\,dw^-\,dw^+= \\
		&= \mu\int_0^1\int_0^1\int_\cO\ave{m^-w^+\frac{\varphi(v')+\varphi(v_\ast')}{2}-\frac{m^+w^-+m^-w^+}{2}\varphi(v)}g(v,w^-,w^+,t)\,dv\,dw^-\,dw^+.
	\end{aligned}
	\label{eq:std_Boltz-AR.2}
\end{align}

To proceed further, we need the following result, which establishes that the average normalised incoming and outgoing degrees are the same, as it is the case for finite graphs.
\begin{lemma}
\label{lemma:m-m+}
It holds that $m^-=m^+$.
\end{lemma}
\begin{proof}
Let us consider at first a finite graph with $N<\infty$ vertices. If $\tilde{w}^-,\,\tilde{w}^+\in\cW=\{\frac{k}{N},\,k=0,\,1,\,\dots,\,N\}$ are the normalised incoming and outgoing degrees of a generic vertex and if $\tilde{m}^-_N$, $\tilde{m}^+_N$ denote the mean normalised incoming and outgoing degrees of the graph then
\begin{align*}
	\tilde{m}^-_N &= \sum_{\tilde{w}^-\in\cW}\sum_{\tilde{w}^+\in\cW}\int_\cO\tilde{w}^-\tilde{g}(v,\tilde{w}^-,\tilde{w}^+,t)\,dv\,\Delta{\tilde{w}}^-\,\Delta{\tilde{w}}^+, \\
	\tilde{m}^+_N &= \sum_{\tilde{w}^+\in\cW}\sum_{\tilde{w}^-\in\cW}\int_\cO\tilde{w}^+\tilde{g}(v,\tilde{w}^-,\tilde{w}^+,t)\,dv\,\Delta{\tilde{w}}^-\,\Delta{\tilde{w}}^+.
\end{align*}
On the other hand, we clearly also have
\begin{align*}
	\tilde{m}^-_N &= \frac{1}{N}\sum_{\tilde{w}^-\in\cW}\tilde{w}^-\,\#\{i\in\cI\,:\,\indeg{i}=N\tilde{w}^-\}=\frac{1}{N^2}\sum_{i\in\cI}\indeg{i}, \\
	\tilde{m}^+_N &= \frac{1}{N}\sum_{\tilde{w}^+\in\cW}\tilde{w}^+\,\#\{i\in\cI\,:\,\outdeg{i}=N\tilde{w}^+\}=\frac{1}{N^2}\sum_{i\in\cI}\outdeg{i},
\end{align*}
where $\#$ denotes the cardinality of a set. Since, in any graph, $\sum_{i\in\cI}\indeg{i}=\sum_{i\in\cI}\outdeg{i}$, we get $\tilde{m}^-_N=\tilde{m}^+_N$, hence
$$ \resizebox{\hsize}{!}{$
	\displaystyle{
	\sum_{\tilde{w}^-\in\cW}\sum_{\tilde{w}^+\in\cW}\int_\cO\tilde{w}^-\tilde{g}(v,\tilde{w}^-,\tilde{w}^+,t)\,dv\,\Delta{\tilde{w}}^-\,\Delta{\tilde{w}}^+
		=\sum_{\tilde{w}^+\in\cW}\sum_{\tilde{w}^-\in\cW}\int_\cO\tilde{w}^+\tilde{g}(v,\tilde{w}^-,\tilde{w}^+,t)\,dv\,\Delta{\tilde{w}}^-\,\Delta{\tilde{w}}^+.
	}
	$} $$
Passing formally to the limit $N\to\infty$ yields then
$$ \int_0^1\int_0^1\int_\cO w^-g(v,w^-,w^+,t)\,dv\,dw^-\,dw^+=\int_0^1\int_0^1\int_\cO w^+g(v,w^-,w^+,t)\,dv\,dw^-\,dw^+, $$
whence, recalling the definitions of $h^-$, $h^+$,
$$ \int_0^1w^-h^-(w^-)\,dw^-=\int_0^1w^+h^+(w^+)\,dw^+ $$
and the thesis follows.
\end{proof}

Owing to Lemma~\ref{lemma:m-m+}, equation~\eqref{eq:std_Boltz-AR.2} becomes
\begin{align*}
	\frac{d}{dt}&\int_0^1\int_0^1\int_\cO\varphi(v)g(v,w^-,w^+,t)\,dv\,dw^-\,dw^+= \\
	&= m\mu\int_0^1\int_0^1\int_\cO\ave{w^+\frac{\varphi(v')+\varphi(v_\ast')}{2}-\frac{w^-+w^+}{2}\varphi(v)}g(v,w^-,w^+,t)\,dv\,dw^-\,dw^+,
\end{align*}
where $m:=m^-=m^+$, which coincides with~\eqref{eq:Boltzmann.g-AR} upon letting $\mu=1/m$.

\begin{figure}
\centering
\includegraphics[width=0.6\linewidth]{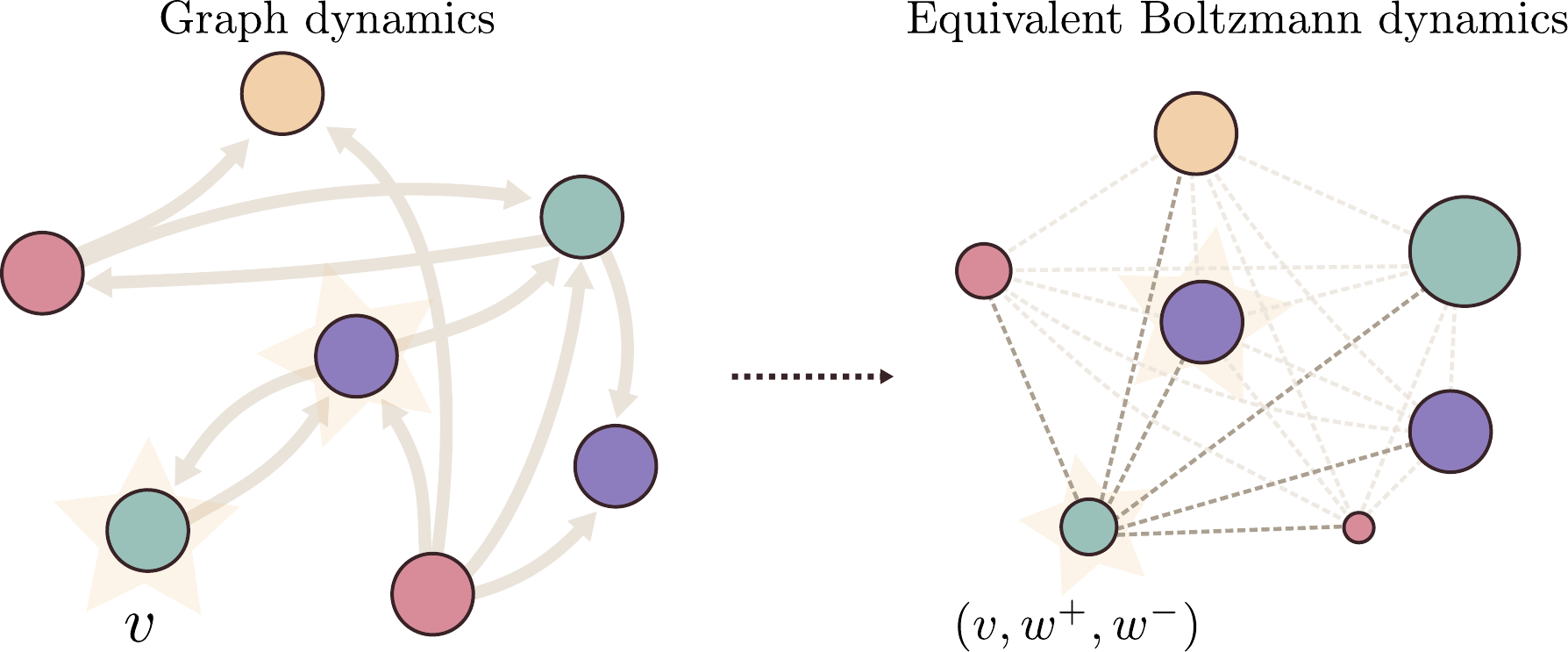}
\caption{Visual representation of the equivalence between the network dynamics and the equivalent Boltzmann one. A situation in which the interactions between the agents are regulated by a network structure (left panel) is replaced, without approximation, by one in which every agent can interact with any other agent (right panel).}
\label{fig:equiv_Boltzmann}
\end{figure}

This shows that, in the case of ``action-reaction'' rules, the collective dynamics emerging from networked interactions may be equivalently described by all-to-all interactions within a Boltzmann-type framework, cf. Figure~\ref{fig:equiv_Boltzmann}, where the interaction rate (viz. collision kernel) is proportional to the product of the incoming and outgoing degrees of the interacting agents, cf.~\eqref{eq:b}. Hence, those dynamics may be faithfully reproduced even if the detailed graph of the agent connections is unknown, provided the statistical distribution of the incoming and outgoing degrees is known.

\subsection{The case of ``action-action'' interactions}
``Action-action'' interactions, cf.~\eqref{eq:int.AA}, are pairwise interactions which, unlike those of the classical kinetic theory, need not produce a simultaneous change of microscopic state of the interacting agents. Technically, the reason is that the two rules in~\eqref{eq:int.AA} feature two different random variables $\Theta$, $\Theta_\ast$ with a law differing essentially in the interaction rate $B$, which might not be symmetric. In order to seek a parallelism with a graph-free Boltzmann-type kinetic description, it is therefore necessary to refer to a generalisation of the Boltzmann-type equation~\eqref{eq:std_Boltz-AR}, which takes into account a possibly non-symmetric collision kernel. Such an equation writes in weak form as
\begin{align}
	\begin{aligned}[b]
		\frac{d}{dt}\int_\cV\Phi(\bv)g(\bv,t)\,d\bv &= \frac{1}{2}\int_{\cV^2}b(\bv,\bv_\ast)\ave{\Phi(\bv')-\Phi(\bv)}g(\bv,t)g(\bv_\ast,t)\,d\bv\,d\bv_\ast \\
   		&\phantom{=} +\frac{1}{2}\int_{\cV^2}b(\bv_\ast,\bv)\ave{\Phi(\bv_\ast')-\Phi(\bv_\ast)}g(\bv,t)g(\bv_\ast,t)\,d\bv\,d\bv_\ast
    \end{aligned}
    \label{eq:std_Boltz-AA}
\end{align}
for every observable $\Phi:\cV\to\R$.

Letting again $\bv=(v,\,w^-,\,w^+)\in\cV=\cO\times [0,\,1]\times [0,\,1]$ and choosing $b$ like in~\eqref{eq:b} we obtain, with $\Phi(\bv)=\varphi(v)$ and after some computations similar to those performed in~\eqref{eq:std_Boltz-AR.1},~\eqref{eq:std_Boltz-AR.2},
\begin{align*}
	&\frac{d}{dt}\int_0^1\int_0^1\int_\cO\varphi(v)g(v,w^-,w^+)\,dv\,dw^-\,dw^+= \\
	&= \mu\int_0^1\int_0^1\int_\cO\ave{\frac{m^-w^+\varphi(v')+m^+w^-\varphi(v_\ast')}{2}-m^-w^+\varphi(v)}g(v,w^-,w^+,t)\,dv\,dw^-\,dw^+.
\end{align*}
Finally, setting $\mu=1/m$ with $m:=m^-=m^+$ (cf. Lemma~\ref{lemma:m-m+}) we recover~\eqref{eq:Boltzmann.g-AA}.

\section{More general interactions}
\label{sect:general_interactions}
It is natural to inquire whether the equivalence between graph-mediated kinetic equations such as~\eqref{eq:Boltzmann.g-AR},~\eqref{eq:Boltzmann.g-AA} and Boltzmann-type equations~\eqref{eq:std_Boltz-AR},~\eqref{eq:std_Boltz-AA} may be extended to non-polarised interactions. In particular, we focus on \textit{linear} interactions:
\begin{equation}
   v'=pv+qv_\ast, \qquad v_\ast'=p_\ast v_\ast+q_\ast v,
   \label{eq:lin_int}
\end{equation}
which provide a sufficiently simple, yet powerful, model for a variety of applications. In~\eqref{eq:lin_int}, $p,\,q,\,p_\ast,\,q_\ast$ are either deterministic or random coefficients chosen so that $v',\,v_\ast'\in\cO$ for all $v,\,v_\ast\in\cO$.

Sticking to the classical spirit of kinetic theory, let us consider ``action-reaction'' interactions, in which to every interaction of the $v$-agent with the $v_\ast$-agent there corresponds a simultaneous interaction of the $v_\ast$-agent with the $v$-agent. The case of ``action-action'' interactions may be treated similarly.

Starting from the graph-mediate kinetic equation~\eqref{eq:F}, by means of a procedure analogous to that of Section~\ref{sect:stat_description} we obtain the following equation:
\begin{align}
	\begin{aligned}[b]
		\frac{d}{dt} &\underbrace{\sum_{\tilde{w}^+\in\cW}\sum_{\tilde{w}^-\in\cW}\int_\cO\varphi(v)\tilde{g}(v,\tilde{w}^-,\tilde{w}^+,t)\,dv\,\Delta{\tilde{w}}^-\,\Delta{\tilde{w}}^+}_{\Circled{\mathrm{I}}} \\
		&= \underbrace{\frac{1}{2N^2}\int_\cO\int_\cO\ave{\varphi(v')+\varphi(v_\ast')}\bf^T(v,t)\bM\bf(v_\ast,t)\,dv\,dv_\ast}_{\Circled{\mathrm{II}}} \\
		&\phantom{=} -\underbrace{\frac{1}{2}\sum_{\tilde{w}^+\in\cW}\sum_{\tilde{w}^-\in\cW}\int_\cO\tilde{w}^+\varphi(v)
			\tilde{g}(v,\tilde{w}^-,\tilde{w}^+,t)\,dv\,\Delta{\tilde{w}}^-\,\Delta{\tilde{w}}^+}_{\Circled{\mathrm{III}}} \\
		&\phantom{=} -\underbrace{\frac{1}{2}\sum_{\tilde{w}_\ast^+\in\cW}\sum_{\tilde{w}_\ast^-\in\cW}\int_\cO\tilde{w}_\ast^-\varphi(v_\ast)
			\tilde{g}(v_\ast,\tilde{w}_\ast^-,\tilde{w}_\ast^+,t)\,dv_\ast\,\Delta{\tilde{w}_\ast}^-\,\Delta{\tilde{w}_\ast}^+}_{\Circled{\mathrm{IV}}}.
	\end{aligned}
	\label{eq:gtilde.lin_int}
\end{align}

While it is clear that, in the limit $N\to\infty$, the terms \Circled{I}, \Circled{III}, \Circled{IV} converge formally to
\begin{align*}
	\Circled{\mathrm{I}} &\to \int_0^1\int_0^1\int_\cO\varphi(v)g(v,w^-,w^+,t)\,dv\,dw^-\,dw^+ \\
	\Circled{\mathrm{III}} &\to \frac{1}{2}\int_0^1\int_0^1\int_\cO w^+\varphi(v)g(v,w^-,w^+,t)\,dv\,dw^-\,dw^+ \\
	\Circled{\mathrm{IV}} &\to \frac{1}{2}\int_0^1\int_0^1\int_\cO w_\ast^-\varphi(v_\ast)g(v_\ast,w_\ast^-,w_\ast^+,t)\,dv_\ast\,dw_\ast^-\,dw_\ast^+,
\end{align*}
where we dropped the tildes for ease of notation, it is much harder to identify the limit, if any, of \Circled{II}, because this term depends on the detailed microscopic couplings of the agents on the graph as encoded in the adjacency matrix $\bM$. In general, one may expect that \Circled{II} cannot be rewritten straightforwardly in terms of the lumped information on the graph connections brought by the incoming and outgoing degrees of the vertices.

Nevertheless, it turns out that for particular choices of the observable $\varphi$ some limit aggregate information may be obtained from~\eqref{eq:gtilde.lin_int}. Let $\varphi$ be linear, then invoking~\eqref{eq:lin_int} we discover:
\begin{align*}
	\Circled{\mathrm{II}} &= \frac{1}{2N^2}\left(\ave{p+q_\ast}\int_\cO\varphi(v)\bf^T(v,t)\bM\ones\,dv+\ave{p_\ast+q}\int_\cO\varphi(v_\ast)\ones^T\bM\bf(v_\ast,t)\,dv_\ast\right)
	\intertext{whence, recalling~\eqref{eq:w-_w+} and~\eqref{eq:w.f},}
	&=\frac{1}{2N}\left(\ave{p+q_\ast}\sum_{w^+=0}^{N}\sum_{w^-=0}^{N}\int_\cO w^+\varphi(v)g_N(v,w^-,w^+,t)\,dv\right. \\
	&\phantom{=} +\left.\ave{p_\ast+q}\sum_{w^+=0}^{N}\sum_{w^-=0}^{N}\int_\cO w^-\varphi(v)g_N(v,w^-,w^+,t)\,dv\right) \\
	&=\sum_{\tilde{w}^+\in\cW}\sum_{\tilde{w}^-\in\cW}\int_\cO\frac{\ave{p+q_\ast}\tilde{w}^++\ave{p_\ast+q}\tilde{w}^-}{2}\varphi(v)
		\tilde{g}(v,\tilde{w}^-,\tilde{w}^+,t)\,dv\,\Delta{\tilde{w}}^-\,\Delta{\tilde{w}}^+ \\
	&\xrightarrow{N\to\infty}\int_0^1\int_0^1\int_\cO\frac{\ave{p+q_\ast}w^++\ave{p_\ast+q}w^-}{2}\varphi(v)g(v,w^-,w^+,t)\,dv\,dw^-\,dw^+.
\end{align*}

Finally, we have obtained that~\eqref{eq:gtilde.lin_int} yields, in the limit $N\to\infty$,
\begin{align}
	\begin{aligned}[b]
		\frac{d}{dt} &\int_0^1\int_0^1\int_\cO\varphi(v)g(v,w^-,w^+,t)\,dv\,dw^-\,dw^+ \\
		&= \int_0^1\int_0^1\int_\cO\frac{\ave{p+q_\ast-1}w^++\ave{p_\ast+q-1}w^-}{2}\varphi(v)g(v,w^-,w^+,t)\,dv\,dw^-\,dw^+
	\end{aligned}
	\label{eq:g.lin_int}
\end{align}
for all \textit{linear} observables $\varphi$. This allows us to recover, in particular, the time trend of the \textit{mean opinion}
$$ \ave{v}(t):=\int_\cO vF(v,t)\,dv=\int_0^1\int_0^1\int_\cO vg(v,w^-,w^+,t)\,dv\,dw^-\,dw^+ $$
by letting $\varphi(v)=v$:
\begin{equation}
	\frac{d}{dt}\ave{v}=\int_0^1\int_0^1\int_\cO\frac{\ave{p+q_\ast-1}w^++\ave{p_\ast+q-1}w^-}{2}vg(v,w^-,w^+,t)\,dv\,dw^-\,dw^+.
	\label{eq:ave.v}
\end{equation}

It can be checked that~\eqref{eq:g.lin_int},~\eqref{eq:ave.v} are the very same equations resulting from the Boltzmann-type equation~\eqref{eq:std_Boltz-AR} in which one uses the collision kernel~\eqref{eq:b} (with $\mu=1/m$) and confines $\varphi$ to linear observables.

Therefore, for linear interaction rules we conclude that:
\begin{enumerate*}[label=(\roman*)]
\item the collective dynamics resulting from a Boltzmann-type approximation of the graph by means of the degree distribution of the agents differ, in general, from those emerging from real graph-mediated interactions, because the latter require the detailed knowledge of the graph connections (cf. the term \Circled{II} in~\eqref{eq:gtilde.lin_int});
\item nevertheless, the expectation of the opinion on the whole graph is correctly reproduced also by a Boltzmann-type approximation, i.e. without the detailed knowledge of the graph connections.
\end{enumerate*}
Although weaker than that obtained in Section~\ref{sect:equiv_boltz}, this result still helps to corroborate the validity of the approaches which approximate the graph structure with the statistical distribution of the degree of the vertices, at least as far as the investigation of certain statistical properties of the system is concerned.

By similar arguments, it is not difficult to see that this equivalence remains valid also for the following generalisation of the interaction rules~\eqref{eq:lin_int}:
\begin{equation}
	v'=p(v)+q(v_\ast), \qquad v_\ast'=p_\ast(v_\ast)+q_\ast(v),
	\label{eq:sep_int}
\end{equation}
where $p,\,q,\,p_\ast,\,q_\ast$ are now either deterministic or random functions defined on $\cO$ and such that $v',\,v_\ast'\in\cO$ for all $v,\,v_\ast\in\cO$. Specifically, from~\eqref{eq:gtilde.lin_int} we deduce that for a linear $\varphi$ it results
$$ \resizebox{\hsize}{!}{$
	\Circled{\mathrm{II}}\xrightarrow{N\to\infty}
		\displaystyle{\int_0^1\int_0^1\int_\cO}
			\dfrac{w^+\ave{\varphi(p(v))+\varphi(q_\ast(v))}+w^-\ave{\varphi(p_\ast(v))+\varphi(q(v))}}{2}	g(v,w^-,w^+,t)\,dv\,dw^-\,dw^+,
	$} $$
whence, passing to the limit $N\to\infty$ in the whole equation,
\begin{align*}
	\frac{d}{dt} &\int_0^1\int_0^1\int_\cO\varphi(v)g(v,w^-,w^+,t)\,dv\,dw^-\,dw^+ \\
	&= \frac{1}{2}\int_0^1\int_0^1\int_\cO w^+\ave{\varphi(p(v))+\varphi(q_\ast(v))-\varphi(v)}g(v,w^-,w^+,t)\,dv\,dw^-\,dw^+ \\
	&\phantom{=} +\frac{1}{2}\int_0^1\int_0^1\int_\cO w^-\ave{\varphi(p_\ast(v))+\varphi(q(v))-\varphi(v)}g(v,w^-,w^+,t)\,dv\,dw^-\,dw^+,
\end{align*}
which holds for all \textit{linear} observables $\varphi$. The same is obtained from the Boltzmann-type equation~\eqref{eq:std_Boltz-AR} with collision kernel~\eqref{eq:b} and $\mu=1/m$. In particular, also in this case the evolution of the mean opinion can be recovered simply from the knowledge of the statistical distribution of the degrees of the vertices of the graph, the detailed structure of graph connections being unnecessary:
\begin{align*}
	\frac{d}{dt}\ave{v} &= \frac{1}{2}\int_0^1\int_0^1\int_\cO w^+\ave{p(v)+q_\ast(v)-v}g(v,w^-,w^+,t)\,dv\,dw^-\,dw^+ \\
	&\phantom{=} +\frac{1}{2}\int_0^1\int_0^1\int_\cO w^-\ave{p_\ast(v)+q(v)-v}g(v,w^-,w^+,t)\,dv\,dw^-\,dw^+.
\end{align*}

Summarising, we have proved the following:
\begin{theorem}\label{theorem:equivalence_separable}
Let the interaction rules be of the form~\eqref{eq:sep_int}. 
Then the evolution of \emph{linear} observables $\varphi$ provided by the graph-mediated kinetic equation~\eqref{eq:F} with whatever adjacency matrix $\bM$ is formally equivalent, in the limit $N\to\infty$, to that provided by the Boltzmann-type equation~\eqref{eq:std_Boltz-AR} with $\mu=1/m$ in~\eqref{eq:b}.
\end{theorem}

\section{On the closure of~\texorpdfstring{\eqref{eq:gtilde.lin_int}}{} in the limit~\texorpdfstring{$\boldsymbol{N\to\infty}$}{}}
\label{sect:approximations_of_M}
As already stated, for non-polarised memory interactions it is hard to identify, in general, the limit of the term \Circled{II} in~\eqref{eq:gtilde.lin_int}. The reason is that if the post-interaction opinions $v'$, $v_\ast'$ are not polarised, i.e. they depend on both pre-interaction opinions $v$, $v_\ast$, then the lumped information brought by the distribution of the incoming and outgoing degrees of the vertices, cf.~\eqref{eq:w-_w+}, may be insufficient to compute \Circled{II}. However, one may expect that~\eqref{eq:gtilde.lin_int} turns into a closed equation for every observable $\varphi$ at least for particular classes of adjacency matrices $\bM$.

\subsection{Complete graphs}
Consider the adjacency matrix
\begin{equation}
	\bM=\ones\ones^T=
	\begin{pmatrix}
		1 & 1 & \cdots & 1 \\
		1 & 1 & \cdots & 1 \\
		\vdots & \vdots & \ddots & \vdots \\
		1 & 1 & \cdots & 1
	\end{pmatrix},
	\label{eq:M.all-to-all}
\end{equation}
which describes a complete graph, i.e. one in which every pair of vertices is connected by a unique edge (self-loops are included) so that agents experience all-to-all interactions. 
In this case, since $w^-=w^+=N$ for every vertex, ~\eqref{eq:g} takes the form
$$ g_N(v,w^-,w^+,t)=F(v,t)\delta_{w^-,N}\delta_{w^+,N}, $$
where $\delta_{\cdot,\cdot}$ denotes the Kronecker delta. Consequently, $\tilde{g}(v,\tilde{w}^-,\tilde{w}^+,t)=N^2F(v,t)\delta_{\tilde{w}^-,1}\delta_{\tilde{w}^+,1}$ and
$$ \Circled{\mathrm{I}}=\Circled{\mathrm{III}}=\Circled{\mathrm{IV}}=\int_\cO\varphi(v)F(v,t)\,dv. $$
On the other hand,
\begin{align*}
	\Circled{\mathrm{II}} &= \frac{1}{2}\int_\cO\int_\cO\ave{\varphi(v')+\varphi(v_\ast')}\frac{1}{N}\left(\ones^T\bf(v,t)\right)^T\frac{1}{N}\left(\ones^T\bf(v_\ast,t)\right)\,dv\,dv_\ast \\
	&= \frac{1}{2}\int_\cO\int_\cO\ave{\varphi(v')+\varphi(v_\ast')}F(v,t)F(v_\ast,t)\,dv\,dv_\ast,
\end{align*}
so that finally~\eqref{eq:gtilde.lin_int} reduces to
$$ \frac{d}{dt}\int_\cO\varphi(v)F(v,t)\,dv=\frac{1}{2}\int_\cO\int_\cO\ave{\varphi(v')+\varphi(v_\ast')-\varphi(v)-\varphi(v_\ast)}F(v,t)F(v_\ast,t)\,dv\,dv_\ast $$
for every $N\in\mathbb{N}$ (in particular, not only in the limit $N\to\infty$), \textit{every} observable $\varphi$ and \textit{every} definition of the interaction rules yielding the post-interaction opinions $v'$, $v_\ast'$. This is the weak form of a classical Boltzmann-type equation with a single microscopic state $v$, which can be equivalently obtained from~\eqref{eq:std_Boltz-AR} with $\cV=\cO$, $\bv=v$ and $b(\bv,\bv_\ast)\equiv 1$. In particular, a constant unitary collision kernel is the counterpart in~\eqref{eq:std_Boltz-AR} of all-to-all interactions expressed by the adjacency matrix~\eqref{eq:M.all-to-all}.

\subsection{Rank-one approximation of~\texorpdfstring{$\bM$}{}}
Let now
\begin{equation}
	\bM\approx\frac{\bw^+(\bw^-)^T}{M_N},
	\label{eq:M.rank-one}
\end{equation}
where $M_N:=\sum_{i\in\cI}\indeg{i}=\sum_{i\in\cI}\outdeg{i}$, which provides a natural rank-one approximation of any adjacency matrix with given incoming and outgoing vertex degrees, cf.~\eqref{eq:w-_w+}. In this case,
\begin{align*}
	\Circled{\mathrm{II}} &= \frac{1}{2M_NN^2}\int_\cO\int_\cO\ave{\varphi(v')+\varphi(v_\ast')}\left((\bw^+)^T\bf(v,t)\right)^T(\bw^-)^T\bf(v_\ast,t)\,dv\,dv_\ast \\
	&= 
		\begin{multlined}[t]
			\frac{1}{2M_N}\sum_{w_\ast^+=0}^{N}\sum_{w_\ast^-=0}^{N}\sum_{w^+=0}^{N}\sum_{w^-=0}^{N}\int_\cO\int_\cO w^+w_\ast^-\ave{\varphi(v')+\varphi(v_\ast')}g_N(v,w^-,w^+,t) \\
			\times g_N(v_\ast,w_\ast^-,w_\ast^+,t)\,dv\,dv_\ast
		\end{multlined}
		\\
	&=
		\begin{multlined}[t]
			\frac{N^2}{2M_N}\sum_{\tilde{w}_\ast^+\in\cW}\sum_{\tilde{w}_\ast^-\in\cW}\sum_{\tilde{w}^+\in\cW}\sum_{\tilde{w}^-\in\cW}\int_\cO\int_\cO\tilde{w}^+\tilde{w}_\ast^-\ave{\varphi(v')+\varphi(v_\ast')}
				\tilde{g}(v,\tilde{w}^-,\tilde{w}^+,t) \\
			\times\tilde{g}(v_\ast,\tilde{w}_\ast^-,\tilde{w}_\ast^+,t)\,dv\,dv_\ast\,\Delta{\tilde{w}^-}\,\Delta{\tilde{w}^+}\,\Delta{\tilde{w}_\ast^-}\,\Delta{\tilde{w}_\ast^+}.
		\end{multlined}
\end{align*}
By inspecting the proof of Lemma~\ref{lemma:m-m+} we notice that
$$ \tilde{m}_N^\pm:=\frac{M_N}{N^2}=\sum_{\tilde{w}^+\in\cW}\sum_{\tilde{w}^-\in\cW}\int_\cO\tilde{w}^\pm\tilde{g}(v,\tilde{w}^-,\tilde{w}^+,t)\,dv\,\Delta{\tilde{w}^-}\,\Delta{\tilde{w}^+} $$
is the mean normalised incoming/outgoing degree of the vertices of the graph, which converges to
$$ m=m^\pm=\int_0^1\int_0^1\int_\cO w^\pm g(v,w^-,w^+,t)\,dv\,dw^-\,dw^+ $$
as $N\to\infty$. Hence formally
\begin{multline*}
	\Circled{\mathrm{II}}\xrightarrow{N\to\infty}\frac{1}{2m}\int_0^1\int_0^1\int_0^1\int_0^1\int_\cO\int_\cO w^+w_\ast^-\ave{\varphi(v')+\varphi(v_\ast')}g(v,w^-,w^+,t) \\
	\times g(v_\ast,w_\ast^-,w_\ast^+)\,dv\,dv_\ast\,dw^-\,dw^+\,dw_\ast^-\,dw_\ast^+,
\end{multline*}
so that, passing to the limit $N\to\infty$ in~\eqref{eq:gtilde.lin_int} with $\bM$ approximated by~\eqref{eq:M.rank-one}, we obtain
\begin{align}
	\begin{aligned}[b]
		\frac{d}{dt} &\int_0^1\int_0^1\int_\cO\varphi(v)g(v,w^-,w^+,t)\,dv\,dw^-\,dw^+ \\
		&=
			\begin{multlined}[t]
				\frac{1}{2}\int_0^1\int_0^1\int_0^1\int_0^1\int_\cO\int_\cO\frac{w^+w_\ast^-}{m}\ave{\varphi(v')+\varphi(v_\ast')-\varphi(v)-\varphi(v_\ast)}g(v,w^-,w^+,t) \\
				\times g(v_\ast,w_\ast^-,w_\ast^+,t)\,dv\,dv_\ast\,dw^-\,dw^+\,dw_\ast^-\,dw_\ast^+.
			\end{multlined}
	\end{aligned}
\end{align}
This corresponds to the general form of the Boltzmann-type equation~\eqref{eq:std_Boltz-AR} with $\cV=\cO\times [0,\,1]\times [0,\,1]$, $\bv=(v,w^-,w^+)$ and the collision kernel $b$ given by~\eqref{eq:b} with $\mu=1/m$, for \textit{every} observable $\varphi$ and \textit{every} interaction rule providing $v'$, $v_\ast'$ as post-interaction opinions.

Therefore, we have proved:
\begin{theorem}
For any type of interaction rule, the graph-mediated kinetic equation~\eqref{eq:F} with the adjacency matrix $\bM$ approximated by the rank-one matrix~\eqref{eq:M.rank-one} is formally equivalent, in the limit $N\to\infty$, to the Boltzmann-type equation~\eqref{eq:std_Boltz-AR} with $\mu=1/m$ in~\eqref{eq:b}.
\end{theorem}

\begin{remark}
We recall that, in the graph-mediated kinetic equations~\eqref{eq:AR-f} and then~\eqref{eq:F}, the entries of the adjacency matrix $\bM$ provide the values of the interaction kernel $B$, cf. Section~\ref{sect:derivation_kin_eq}. Thus, they represent the interaction rates of pairs of agents/vertices of the graph. The rank-one approximation~\eqref{eq:M.rank-one} of $\bM$ corresponds to assuming that such rates are simply proportional to the gross incoming and outgoing degrees of the agents/vertices regardless of the detailed graph topology, as if agents were substantially ``independent'' from the point of view of the graph connections. Remarkably, this independence is the key to the closure of~\eqref{eq:F} in the sense of a classical Boltzmann-type description.
\end{remark}

\section{Numerical experiments}
\label{sect:numerics}
In order to validate the results obtained in the previous sections, we perform a series of numerical experiments choosing as the underlying graph a real social network built from the ``Social circles: Twitter'' dataset~\cite{snapnets,Twitter_Dataset}. This dataset contains 81\,306 vertices, viz. users, and 1\,768\,149 edges, viz. their social connections.

\begin{algorithm}[!t]
\caption{Monte Carlo algorithm for ``action-reaction'' Boltzmann-type equations on a graph }\label{alg:monte_carlo}
\begin{algorithmic}
\Require adjacency matrix $\bM$; initial state $V_0\in\cO^N$; time step $\Delta{t}>0$; final time $T>0$
\State $\tilde{V}\gets V_0$
\State $t\gets 1$
\For{$t<T$}
\State $\ave{\varphi}(t)\gets\frac{1}{N}\sum_{i=1}^N\varphi(\tilde{V}(i))$ 
\State $V\gets\tilde{V}$
\State $P\gets \text{random permutation of } \left\{1,\,\dots,\,N\right\}$
\State $p_1\gets (P(1),\,\dots,\,P(N/2))$
\State $p_2\gets (P(N/2 + 1),\,\dots,\,P(N))$
\State $i\gets 1$
\For{$i<N/2$}
\State $\Theta\sim\operatorname{Bernoulli}\!\left(B(p_1(i),p_2(i))\Delta{t}\right)$
\State $\tilde{V}(p_1(i))\gets V(p_1(i))(1-\Theta)+\Psi(V(p_1(i)),V(p_2(i)))\Theta$
\State $\tilde{V}(p_2(i))\gets V(p_2(i))(1-\Theta)+\Psi_\ast(V(p_2(i)),V(p_1(i)))\Theta$
\State $i\gets i+1$
\EndFor
\State $t\gets t+\Delta t$
\EndFor
\end{algorithmic}
\end{algorithm}

For the sake of completeness, we quickly recall how to solve approximately Boltzmann-type equations by a Monte Carlo numerical approach. Algorithm~\ref{alg:monte_carlo} consists in simulating literally the binary interaction dynamics described by \eqref{eq:interaction_algorithm}. At each time step, agents are randomly paired and, for each pair, a Bernoulli random variable depending on the interaction kernel $B$ is sampled to determine whether an interaction occurs. If it does then the states are updated according to the corresponding interaction functions. In the case of the graph-mediated kinetic equation \eqref{eq:F} the interaction kernel $B$ depends on the adjacency matrix $\bM$. Conversely, in the case of the Boltzmann-type equation~\eqref{eq:std_Boltz-AR} it depends only on the incoming and outgoing degrees as specified in~\eqref{eq:b}.

\begin{figure}[!t]
\centering
\includegraphics[width=\linewidth]{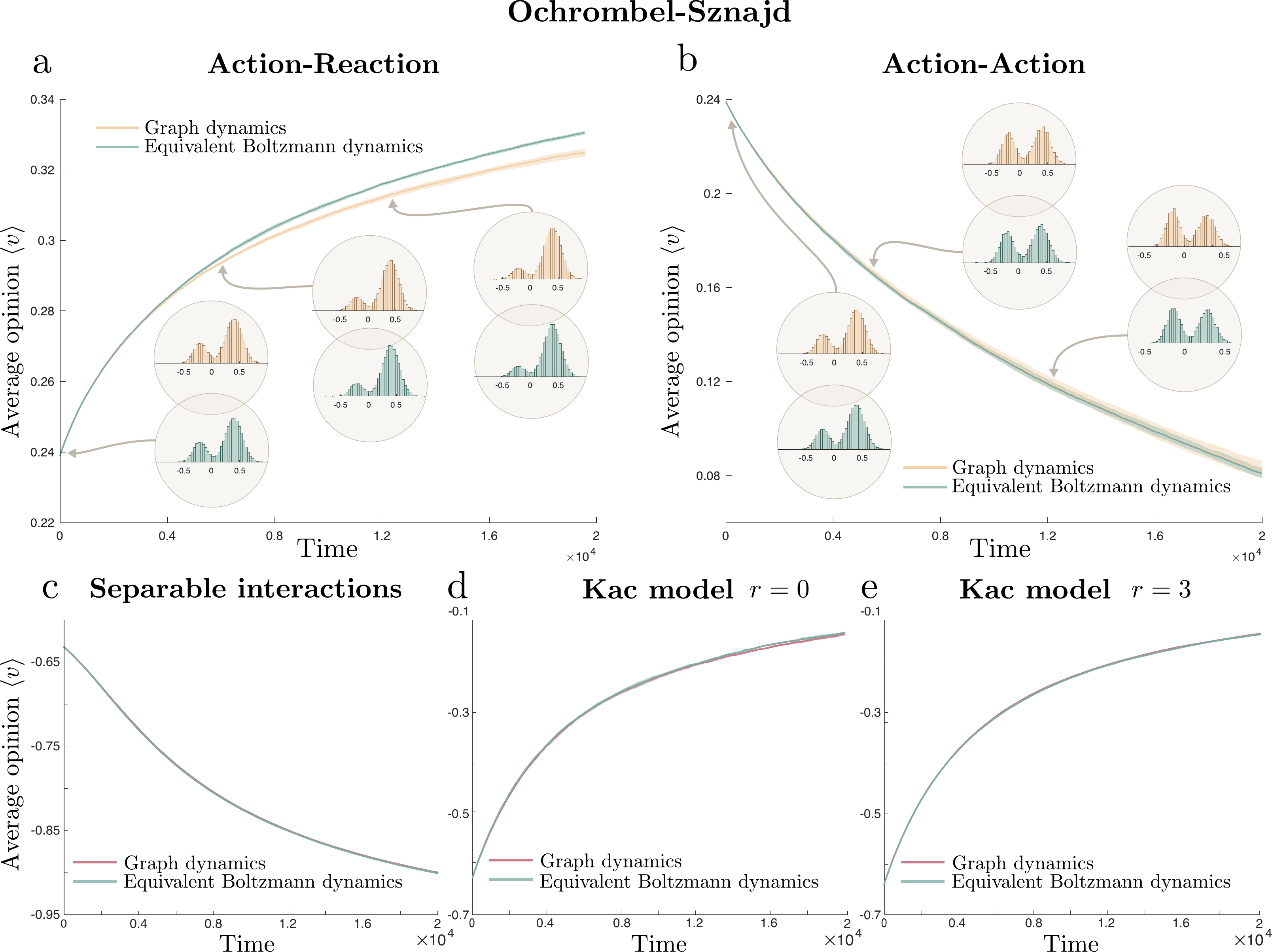}
\caption{Numerical validation of the equivalence between the graph-mediated kinetic equations~\eqref{eq:AR-f},~\eqref{eq:AA-f} and the equivalent Boltzmann-type equations~\eqref{eq:std_Boltz-AR},~\eqref{eq:std_Boltz-AA} using the ``Social circles: Twitter'' network dataset~\cite{snapnets,Twitter_Dataset} for different interaction rules. Panels \textbf{a}, \textbf{b}: ``action-reaction'' and ``action-action'' Ochrombel opinion formation model, cf. Example~\ref{ex:Ochrombel}. The solid lines and filled regions represent respectively mean values and 95\% confidence intervals computed over 10 repetitions. Both dynamics share the same initial condition. Panel \textbf{c}: separable interaction rule~\eqref{eq:separable}. 
Panels \textbf{d}, \textbf{e}: inelastic Kac-inspired separable interaction rule~\eqref{eq:Kac}.}
\label{fig:Twitter}
\end{figure}

As a first experiment, we check the equivalence between graph-mediated and Boltzmann-type kinetic descriptions discussed in Section~\ref{sect:equiv_boltz}, i.e. in the case of polarised memory interactions. 
For simplicity, we choose as interaction rule the Ochrombel simplification of the Sznajd opinion formation model, cf. Example~\ref{ex:sznajd}. We simulate both the ``action-reaction'' (Figure~\ref{fig:Twitter}a) and the ``action-action'' (Figure~\ref{fig:Twitter}b) versions of the dynamics, first on the true network with the graph-mediated kinetic equations~\eqref{eq:AR-f} and~\eqref{eq:AA-f}, respectively, then with their equivalent Boltzmann-type equations~\eqref{eq:std_Boltz-AR},~\eqref{eq:std_Boltz-AA}. We run the simulations for $20\,000$ time steps, averaging over $10$ repetitions. In Figure~\ref{fig:Twitter}a and Figure~\ref{fig:Twitter}b we plot the time evolution of the mean opinion resulting from the graph-mediated kinetic equation (orange line) and from the equivalent Boltzmann-type equation (blue line). We stress that the latter is not aware of the graph structure but only of the distributions of the incoming and outgoing degrees. 
In the same panels we also plot the statistical distribution of the opinion at three different time instants. As predicted by our theory, the Boltzmann-type equation reproduces faithfully the trends of the graph-mediated kinetic equation. In particular, we notice that the trends of the mean opinion are opposite in the ``action-reaction'' and ``action-action'' dynamics. This is in agreement with the respective equivalent Boltzmann-type equations~\eqref{eq:std_Boltz-AR}-\eqref{eq:b} and~\eqref{eq:b}-\eqref{eq:std_Boltz-AA}, whose right-hand sides turn out to differ only in the signs, which are indeed opposite, when they are evaluated with the Ochrombel interaction rules of Example~\ref{ex:sznajd}. Empirically, this can be explained by observing that in the ``action-reaction'' dynamics agents with a higher \textit{incoming} degree are more likely to change opinion. Conversely, in the ``action-action'' dynamics agents with a higher \textit{outgoing} degree are more likely to change opinion. Therefore, given the same initial joint opinion-degrees distribution, the trends of the mean opinion are expected to be opposite.

When the interactions have the more general separable form~\eqref{eq:sep_int}, Theorem \ref{theorem:equivalence_separable} ensures that the equivalence still holds for the mean opinion. We validate this result for the following separable interaction rule (Figure~\ref{fig:Twitter}c)
\begin{equation}
    p(v)=p_\ast(v)=\operatorname{sign}(v)\min\left\{\frac{1}{2},\,v^2\right\}, \qquad
    q(v)=q_\ast(v)=\operatorname{sign}(v)\min\left\{\frac{1}{2},\,\sqrt{\abs{v}}\right\}
    \label{eq:separable}
\end{equation}
and for the inelastic Kac-inspired interaction rule
\begin{equation}
    p(v)=p_\ast(v)=\cos(\theta)\abs{\cos(\theta)}^rv, \qquad
    q(v)=-q_\ast(v)=-\sin(\theta)\abs{\sin(\theta)}^rv,
    \label{eq:Kac}
\end{equation}
cf.~\cite{pulvirenti2004asymptotic}, where $\theta\sim\mathcal{U}([0,\,2\pi])$ is a uniformly distributed random parameter and the exponent $r$ is chosen to be either $r=0$ (Figure~\ref{fig:Twitter}d) or $r=3$ (Figure~\ref{fig:Twitter}e). In all cases, the trends of the mean opinion yielded by the graph-mediated kinetic equation and by the Boltzmann-type equation can be seen to coincide as predicted by the theory.

\begin{remark}
The choice~\eqref{eq:Kac} corresponds precisely to the inelastic Kac model described in~\cite{pulvirenti2004asymptotic}, where however the variable $v$ does not stand for the opinion of social network users but for the speed of gas molecules. Indeed, the Kac model is supposed to represent a caricature of a one-dimensional gas, whose molecules undergo a mixing of their speeds rather than proper physical collisions so as to give rise to non-trivial one-dimensional dynamics. In our context, we may interpret a graph-based Kac model as a one-dimensional caricature of a gas in which molecule ``collisions'' are \textit{heterogeneously} distributed. In particular, molecules with a low (resp. high) outgoing degree ``hit'' few (resp. many) other molecules while molecules with a low (resp. high) incoming degree are ``hit'' by few (resp. many) other molecules.
\end{remark}

\section{Conclusions}
\label{sect:conclusions}
In this paper, we have considered the problem of modelling networked interacting multi-agent systems by means of kinetic equations incorporating a statistical description of the graph of connections among the agents. The main goal was to obtain evolution equations which, in the spirit of the kinetic theory, did not require a detailed knowledge of the graph topology while still retaining fundamental features of the possibly inhomogeneous distribution of the connections.

Starting from networked particle interaction models, first we have derived graph-mediated kinetic equations. By this we mean a system of Boltzmann-type equations which describe the evolution of the state of each vertex, viz. agent, of the graph and are coupled according to the precise structure of the graph connections encoded in the adjacency matrix of the graph. Next, we have shown that, formally, in the limit of an infinite number of vertices of the graph such a system can be reduced to a single Boltzmann-type equation defined on an augmented state space, namely one which includes also the (normalised) incoming and outgoing degrees of the agents regarded as continuous variables in the interval $[0,\,1]$. In particular, the limit procedure has allowed us to identify a precise expression of the interaction kernel of such a Boltzmann-type equation, which carries all the information about the aforesaid degrees thereby constituting a statistical approximation of the adjacency matrix. Remarkably, such an expression, which is proportional to the product of the (normalised) incoming and outgoing degrees of the interacting agents, matches consistently the one postulated heuristically in the literature via the concept of connectivity of the agents. In conclusion, we have proved that a networked interacting particle system can be described statistically by a Boltzmann-type equation, whose interaction kernel replaces the adjacency matrix of the graph and where the kinetic distribution function does not only account for the distribution of the particle state involved in the binary interactions but also for the statistical distribution of the degrees of the graph.

As a matter of fact, this result is exact only for a very special class of binary interactions, that we have called polarised interactions. These are interactions in which the post-interaction states depend only on one of the two pre-interaction states. For more general interactions, such as linear or separable interactions, the limiting Boltzmann-type equation is not equivalent to the original system of graph-mediated equations, meaning that the sole degree distribution is not sufficient to reproduce faithfully the collective dynamics on the graph. Nevertheless, in this case we have proved that the limit equation yields the correct time evolution of the mean state of the agents. In more generality, we have also proved that the limit equation describes correctly the collective evolution of any networked interacting particle system whose adjacency matrix is approximated \textit{a priori} by a rank-one matrix obtained from the product of the incoming and outgoing degree vectors of the graph. This provides a powerful strategy upon which to rely in practical cases in which one is interested in the collective trend of processes taking place on large, complex networks, for which the knowledge of the detailed structure may be challenging or even impossible. Indeed, the degree distributions of a graph are often easier to obtain than the full adjacency matrix.

Further research directions may consist in detailed analytical investigations of the limit Boltz\-mann-type equations to ascertain the impact of different degree distributions on the marginal distribution of the physical state variable of the particles as well as in the extension of the present study to the case of ensembles of graphs.

\section*{Acknowledgments}
A.T. is member of GNFM (Gruppo Nazionale per la Fisica Matematica) of INdAM (Istituto Nazionale di Alta Matematica ``F. Severi''), Italy.
\paragraph{Competing interests:} The authors declare none.

\bibliographystyle{plain}
\bibliography{NmRmTa-Boltzmann_graphs}
\end{document}